\documentclass[12pt]{article}
\usepackage[T2A]{fontenc}
\usepackage[utf8]{inputenc}
\usepackage[russian,english]{babel}
\usepackage{amsmath,amsfonts,amssymb,amsthm}
\usepackage{cite}
\usepackage{graphicx}
\usepackage{epstopdf}
\usepackage{hyperref}
\hypersetup{
colorlinks=true,
linkcolor=red,
citecolor=blue,
urlcolor=blue
}

\allowdisplaybreaks

\usepackage[a4paper, margin = 1in]{geometry}

\newtheorem{theorem}{Theorem}
\newtheorem{proposition}{Proposition}
\newtheorem{remark}{Remark}
\newtheorem{definition}{Definition}

\def\Tr{\mathrm{Tr}}

\title{\bf 
Phenomenon of a stronger trapping behavior in $\Lambda$-type quantum systems with symmetry}
\author{
Boris Volkov, Anastasia Myachkova and Alexander Pechen*\\[0.5em]
\small Department of Mathematical Methods for Quantum Technologies,\\
\small Steklov Mathematical Institute of Russian Academy of Sciences,\\
\small and University of Science and Technology MISIS
}
\date{\today}

\begin{document}
\maketitle
\begingroup
\renewcommand{\thefootnote}{}
\footnotetext{\texttt{*apechen@gmail.com}}
\endgroup

\begin{abstract} 
$\Lambda$, $V$, $\Xi$ (ladder), and other three-level quantum systems with one forbidden transition (referred here as $\Lambda$-type systems) play an important role in quantum physics. Various applications require manipulation of such systems using as control shaped laser field. In this work, we study how degeneracy of energy states or of Bohr frequencies in these systems affects the efficiency or difficulty of finding optimal shape of the control field. For this, we adopt the notion of higher order traps, which was introduced in [A.N. Pechen and D.J. Tannor, Are there traps in quantum control landscapes? Phys. Rev. Lett. {\bf 106}, 120402 (2011)], where second/third order traps were discovered for $\Lambda$-atom with one forbidden transition and with non-degenerate energy levels. We theoretically study control of such systems with and without degeneracy in their eigenstates and Bohr frequencies, and investigate numerically using GRAPE and l-BFGS algorithms how this degeneracy influences on the efficiency of optimizing the control laser field. We find that the degeneracy of the Bohr frequencies in the $\Xi$ system, which makes the system energy levels symmetrically distributed, leads to the appearance of a seventh order trap with a more significant attracting domain resulting in a more difficult optimization, while the degeneracy of energy states in generic $\Lambda$-type systems does not lead to an increase of the order of the zero control trap compared to the non-degenerate case. We also find that when not only the Bohr frequencies are degenerate in the system $\Xi$, but also the dipole moments for the two allowed transitions coincide (in this case $\Xi$ system is not controllable), then true traps arise in the quantum control landscape. In particular, the constant zero control becomes a trap.
\end{abstract}


\section{Introduction}

Quantum control is a wide direction in quantum physics with various existing and prospective applications in quantum technologies~\cite{KochEPJ2022,Brif2012,TannorBook2007,Gough2012,Ansel_Dionis_Arrouas_Peaudecerf_Guerin_Guery-Odelin_Sugny_2024}. An important problem in quantum control is the analysis of quantum control landscapes, which allows to determine the level of difficulty of finding controls for optimal manipulation of quantum systems in numerical or laboratory experiments. 
Quantum control landscapes have found a variety of applications in physics and chemistry. They were exploited to manipulate the intensity of the Autler-Townes components in the photoelectron spectrum~\cite{Wollenhaupt_Prakelt_Sarpe-Tudoran_Liese_Baumert_2005}, making an experimental implementation for retinal photoisomerization in bacteriorhodopsin~\cite{Marquetand_Nuernberger_Vogt_Brixner_Engel_2007}, manipulating molecular systems~\cite{Ruetzel_Stolzenberger_Fechner_Dimler_Brixner_Tannor_2010}, inducing multi-photon excitations in atoms and vibrational population transfer in molecules~\cite{Palao_Reich_Koch_2013}, discovering the failure of greedy algorithms to generate fast quantum gates~\cite{Zahedinejad_Schirmer_Sanders_2014}, experimentally observing saddle points~\cite{Sun_Pelczer_Riviello_Wu_Rabitz_2015} and analyzing quantum state preparation and entanglement creation~\cite{Li_2023} in two spin quantum systems, dressing chopped-random-basis optimization~\cite{Rach_Muller_Calarco_Montangero_2015}, discovering a discontinous phase transition with broken symmetry in a two-qubit quantum system~\cite{Bukov_Day_Weinberg_Polkovnikov_Mehta_Sels_2018}, etc.

Theoretical analysis of quantum control landscapes takes its origin from the fundamental work~\cite{RHR}, where control landscapes for average values of quantum observables of $N$--level quantum systems were analyzed in the kinematic representation, when the controls are considered as matrix elements of the unitary evolution operator. Ultimately important for practical applications is the analysis of quantum control landscapes in the dynamical representation, when the control is the temporal shape of the laser field in the system Hamiltonian~\cite{HoRabitz2006}. Despite large efforts, a rigorous analysis of the control landscape in the dynamical representation, including proof of existence or absence of traps, has remained a challenge which is not yet fully solved even now. The absence of traps was rigorously proven only for two-level quantum systems ($N=2$)~\cite{PechenPRA2012,VolkovJPA2021,VolkovPechenIzv}. Traps were discovered for some systems with $N\geq 4$ and for not very large control times~\cite{deFouquieresSchirmer} (for larger times these traps disappear). Kinematic control landscape for the trace fidelity objective function defined on the $\operatorname{SU}(N)$ group for $N\ge 5$ was shown to have kinematic local extrema which are not global extrema~\cite{Birtea_Casu_Comanescu_2022}. However, similar kinematic traps on $\operatorname{SU}(N)$ were shown to not arise for a different choice of the quantum gate objective functional~\cite{deFouquieresSchirmer}. The landscape structure for various systems was studied using gradient-based optimized trajectories~\cite{Nanduri_Donovan_Ho_Rabitz_2013} and following exactly straight trajectories in the control space~\cite{Nanduri_Ho_Rabitz_2016}, it was exploited for uncertain or unknown quantum control systems~\cite{Wu_Ding_Dong_Wang_2019} with unsupervised machine learning~\cite{Niu_Boixo_Smelyanskiy_Neven_2019}, deep learning based on Alpha-Zero~\cite{Dalgaard_Motzoi_Sherson_2022}, etc. In numerical simulations, possible existence of multiple local minima was revealed for time-minimum control of two-level open quantum systems with unbounded controls entering only in the Hamiltonian~\cite{Clark_Bloch_Colombo_Rooney_2020} or in both Hamiltonian and dissipator~\cite{PetruhanovPechenPhotonics2023}.

In~\cite{PechenTannor2011,PechenTannor2012}, examples of {\it third order traps} were constructed for special $N$--level quantum systems with $N\ge 3$. In particular, a $\Lambda$-atom which is a three-level quantum system with forbidden direct transition between the ground and the intermediate states [Fig.~\ref{Fig1:structure}(a)] was analyzed. In~\cite{VolkovPechen}, examples of traps of arbitrarily high order were constructed for some degenerate quantum systems. Second order traps were found for a $\Lambda$-system~\cite{PechenTannor2011,PechenTannorReply,PechenTannor2012} based on the analysis of the Taylor expansion of the objective functional up to the third order. In~\cite{PechenTannor2011,PechenTannorReply,PechenTannor2012}, this system was analyzed only for the anharmonic case, i.e. when the transition (Bohr) frequencies between the two pairs of interacting states are different. The harmonic case, when these two transition frequencies are the same, was not considered.  
\begin{figure*}[t]
\includegraphics[width=\linewidth]{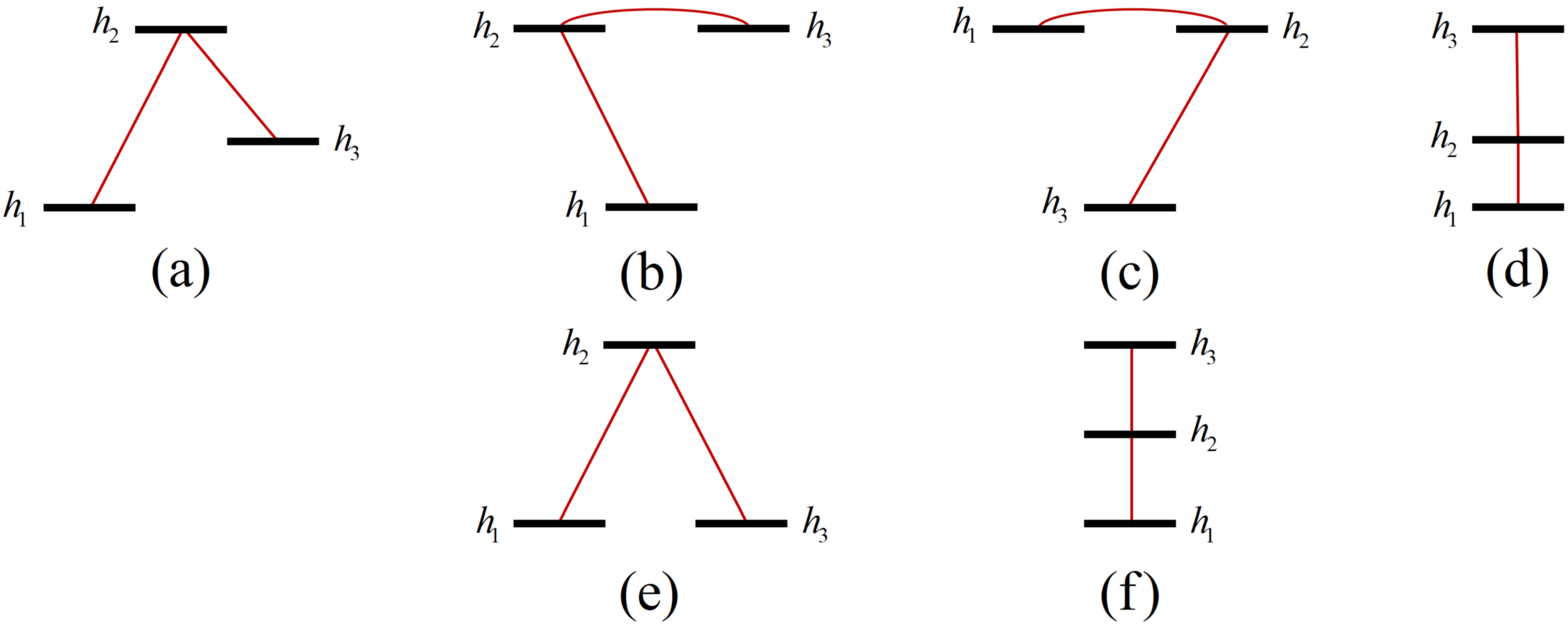}
\caption{Energy level structure of three-level $\Lambda$-type quantum systems with two allowed (and one forbidden) transitions with frequencies $\omega_1=h_2-h_1$ and $\omega_2=h_3-h_2$. Upper row shows the energy level structure for the case $\left | \omega_1 \right |\neq\left | \omega_2 \right |$: (a) $\Lambda$-atom; (b) and (c) degenerate system with two states of the same energy; (d) non-degenerate $\Xi$ (Ladder) system. Lower row shows the energy level structure for the case $|\omega_1|=|\omega_2|$: (e) degenerate~$\Lambda$-atom with two states of the same energy; (f) degenerate $\Xi$-system with two equal transition frequencies. Systems (a-e) are anharmonic systems, system (f) is a harmonic system. The $V$-system is not shown, since for our analysis it is equivalent to the $\Lambda$-atom (in general, systems with energy level structure reflected with respect to the horizontal line are equivalent for our analysis). The ladder system~(f) has symmetrically arranged energy states and when controllable (i.e., when the absolute values of the dipole moments for the two allowed transitions are different) it has trap of the seventh order. Other systems (a-e) have traps of the third order.}
\label{Fig1:structure}
\end{figure*}

As we show in the present work, this case with degenerate transition frequencies is essentially different from the non-degenerate case. For generality, we consider $\Lambda$ (Lambda), $V$ (Vee), $\Xi$ (Ladder), and other three-level quantum systems with one forbidden transition. Such $\Lambda$~\cite{RefLambda}, $V$~\cite{RefV}, and $\Xi$~\cite{RefLadder1,RefLadder2} systems play an important role in quantum physics, where they appear in modeling electromagnetic transparency~\cite{KocharovskayaKhanin1988,Harris1989}, describe spin-one like quantum systems, etc.~\cite{Scully_Zubairy_1997}, etc. Various control problems were studied for these systems, such as minimizing an objective functional which expresses a compromise between the energy of the control and the average population in the highest level~\cite{Alessandro_Isik_2024}, etc. 

We describe the quantum control landscapes for these systems controlled by a time-shaped laser field, and analyze how degeneracy of the Bohr frequencies affects the efficiency of finding the optimal shape of the control laser field. For this, we slightly modify the notion of higher order traps which was introduced in~\cite{PechenTannor2011,PechenTannor2012} and show that the degeneracy of the transition frequencies (such that the energy states are symmetrically arranged) leads to the appearance of a seventh order trap with a much more significant trapping behavior resulting in a more difficult optimization. Both harmonic and anharmonic cases are studied and expressions for subsets of critical controls are obtained. For the uncontrollable $\Xi$ system, we find true traps in the quantum control landscape. These findings allow to more deeply understand the influence of the symmetry and degeneracy of the system Hamiltonian or of the Bohr frequencies on the underlying quantum control landscapes. 

The structure of this work is the following. In Sec.~\ref{Sec:Hamiltonian}, the three-level systems with one forbidden transitions are described and the control problem is formulated. In Sec.~\ref{Sec:Controllability}, controllability of these systems is discussed. In Sec.~\ref{Sec:Landscape} our main result about presence of seventh order traps is formulated. Its proof is summarized in Appendices~\ref{A},~\ref{B}, and~\ref{C}. In Sec.~\ref{uncontrollable_Xi_system} existence of traps for the harmonic uncontrollable $\Xi$ system is formulated with details of the proof provided in Appendix~\ref{D}. In Sec.~\ref{Sec:Numerical} a comparison of the influence of the third and seventh order traps on the efficiency of the optimization is investigated using GRAPE approach~\cite{GRAPE, GRAPE_2}. Conclusions Sec.~\ref{Sec:Conclusions} summarizes the results.

\section{Hamiltonian of the three-level systems with one forbidden transitions}\label{Sec:Hamiltonian}

We consider a three-level quantum system whose dynamics is governed by the Schr\"odinger equation with time-dependent Hamiltonian:
\begin{equation}\label{Eq:Shred}
i\frac{dU_t^f}{dt}=(H_0+f(t)V)U_t^f,\quad U_{0}^f=\mathbb I.
\end{equation}
Here $H_0$ and $V$ are the free and interaction Hamiltonians ($3\times 3$ Hermitian matrices) and $f$ is a coherent control (shaped laser field). We consider the case when transition between one pair of energy states is forbidden. Then, with the appropriate
choice of basis the free and the interaction Hamiltonians can be written as $3\times 3$ matrices
\begin{equation}
\label{Lambda-atom}
H_0= 
\begin{pmatrix}
h_1 & 0 & 0\\
0 & h_2 & 0\\
0 & 0 & h_3
\end{pmatrix},\quad 
V=\begin{pmatrix}
0& v_{12}& 0\\
v^\ast_{12}& 0&v_{23}\\
0&\;\,\, v^\ast_{23}&\;\,\, 0
\end{pmatrix},
\end{equation}
where $h_i\in\mathbb R$ and $v_{12},v_{23}\in\mathbb C\setminus\{0\}$. For a technical convenience, we consider states $|1\rangle$ and $|3\rangle$ as the pair of states with the forbidden direct transition. We do not assume that $h_1, h_2, h_3$ arranged in ascending (or descending) order. 

Various cases of such systems are shown on Fig.~\ref{Fig1:structure}. Strictly speaking, such systems not always represent a $\Lambda$-atom ($\Lambda$-system), for which the direct transition between the two lowest energy states is forbidden and which is shown in Fig.~\ref{Fig1:structure}(a) and Fig.~\ref{Fig1:structure}(e). However, for the sake of language we will refer to all such systems as~$\Lambda$-type systems. The two allowed transitions have frequencies $\omega_1=h_2-h_1$ and $\omega_2=h_3-h_2$. The system is called harmonic if $\omega_1=\omega_2=:\omega$ [this is the case of the degenerate $\Xi$-system which is shown in Fig.~\ref{Fig1:structure}(f)], so that it has an equidistant spectrum similar to harmonic oscillator, and anharmonic if $\omega_1\neq \omega_2$ (see~\cite{SchirmerFuSolomon}). We consider $f\in L_2([0,T],\mathbb{R})$, since in this case the space of controls is a Hilbert space. Using such space of controls is convenient for Hessian analysis, and also since $L_2([0,T],\mathbb{R})\subset L_1([0,T],\mathbb{R})$, so that by Carath\'eodory's existence theorem~\cite{Filippov} the Schr\"odinger equation~(\ref{Eq:Shred}) for every control $f$ has a unique absolutely continuous solution. 

We consider a control problem of maximizing the Mayer type quantum control objective functional representing average value of a quantum Hermitian observable $O$ ($3\times 3$ matrix). This functional has the form
$$
J_O(f)=\Tr \left[U_T^f \rho_0 {(U_T^f)}^{\dagger}O\right],
$$
where $\rho_0$ is the initial system density matrix. To construct a quantum control landscape with trapping behaviour, we consider diagonal quantum observables $O={\rm diag}(\lambda_1,\lambda_2,\lambda_3)$ and assume $\lambda_1>\lambda_3>\lambda_2$. The initial system state generally can be either a pure state or a mixed density matrix $\rho_0$ (a positive operator with unit trace, $\rho_0\geq 0$, $\mathrm{Tr}\rho_0=1$). In this work, we consider $\rho_0=|3\rangle \langle 3|$.

\section{Controllability of the three-level systems with one forbidden transition}\label{Sec:Controllability}

An important step before the analysis of quantum control landscapes is to establish the degree of controllability of the controlled quantum system. Such degree determines upper and lower bounds in the control landscape, i.e. maximal and minimal attainable values of the objective functional. In this section, we briefly overview controllability of the considered systems. While there is a number of controllability notions~\cite{AlbertiniDDAlessandro}, we focus on the notion which is usually referred to as operator (or complete) controllability~\cite{Ramakrishna, SchirmerFuSolomon}. It means that a closed quantum system is controllable if there exist a time $T_{\mathrm{min}}>0$ such that for any time $T>T_{\mathrm{min}}$ and any unitary evolution $ U \in \operatorname{U}(N)$ there exists an admissible control $f\in L_2([0,T],\mathbb{R})$ which implements this evolution as a solution of the Schr\"odinger equation~\eqref{Eq:Shred}, up to a global (physically non-relevant) phase factor $e^{i \varphi}$ with $\varphi\in [0,2\pi)$, such that $U = e^{i \varphi} U_{T}^f$. Various research~\cite{TuriniciRabitz, PolackSuchowskiTannor, SchirmerFuSolomon} study the controllability for the systems considered in our work for some classes of matrix elements in the interaction Hamiltonian $V$. Below we remind the controllability classification, as summarized in~\cite{KuznetsovPechen_2023}, for arbitrary non-zero complex values of matrix elements of $V$:
\begin{itemize}
    \item Systems with $ |\omega_{1}| \neq |\omega_{2}| $ [shown in Fig.~\ref{Fig1:structure}(a-d)]  are \textbf{always} controllable;
    \item Systems with $ \omega_{1} = -\omega_{2} $ [shown in Fig.~\ref{Fig1:structure}(e)]  are \textbf{always not} controllable;
    \item Systems with $ \omega_{1} = \omega_{2}\neq 0$ [shown in Fig.~\ref{Fig1:structure}(f)] are controllable \textbf{iff} $ |v_{12}| \neq |v_{23}| $.
\end{itemize}

\section{Trapping features for the three-level systems with one forbidden transition}\label{Sec:Landscape}

In this section, we provide our main result stating that in difference to the anharmonic three-level quantum system with one forbidden transition which has a third order trap, the (degenerate) harmonic system has a seventh order trap.
\begin{definition}\label{DefTraps}
   A control $f\in L_2([0,T],\mathbb{R})$ is a trap if $f$ is a point of local maximum of the control objective functional $J_O$ but not a point of global maximum of $J_O$. 
\end{definition}
From practical point of view, traps are controls which slow down the search for the globally optimal controls by local search algorithms (e.g., by gradient ascent). 

Controls which could slow down operation of local search first order algorithms can be more general than traps in Definition~\ref{DefTraps}. As an example of such controls, higher order traps were defined in~\cite{PechenTannor2012}. Here we use a slightly modified definition~\cite{VolkovPechen}.
\begin{definition}
\label{maindefinition}
For a given objective functional $J_O$, {\it a trap of $n$-th order} is a control $f\in L_2([0,T],\mathbb{R})$ which satisfies the following two conditions:
\begin{enumerate}
\item $f$ is not a point of global maximum of $J_O$;
\item the Taylor expansion of the objective functional at the point $f$ has the form
\begin{eqnarray*}
J_O(f+\delta f)=J_O(f)+\sum\limits_{j=2}^{n} \frac 1{j!}J_O^{(j)}(f)(\delta f,\ldots,\delta f) +o(\|\delta f\|^{n}) \text{\;as $\|\delta f\|\rightarrow 0$,}
\end{eqnarray*}
where the functional 
$$R(\delta f):=\sum\limits_{j=2}^{n} \frac 1{j!}J_O^{(j)}(f)(\delta f,\ldots,\delta f)$$ is such that:
\begin{enumerate}
\item $R\neq 0$;
\item for any $\delta f\in L_2([0,T],\mathbb{R})$ there exists $\varepsilon>0$ such that
$R(t\delta f)\leq 0$ for all $t\in (-\varepsilon,\varepsilon)$.
\end{enumerate}
\end{enumerate}
\end{definition}

\begin{remark}
In the definition of the $n$-th order trap in~\cite{PechenTannor2012}, instead of the condition b), a stronger condition that $R(\delta f)\leq 0$ for all $\delta f \in L_2([0,T],\mathbb{R})$ was used.
\end{remark}

According to Definition~\ref{maindefinition}, traps of the $n$-th order are critical points of the objective functional. In particular, critical points of the objective functional at which the Hessian is semi-definite and is non-zero are traps of the second order.

We consider the trapping behavior in the vicinity of the zero (null) constant control $f_0\equiv0$. Recall that a control $f\in L_2([0,T],\mathbb{R})$ is called {\it regular} if the first Fr\'echet differential $U_T'(f)$ of the evolution mapping $U_T$ at $f$ is surjective; otherwise the control is called {\it singular}. It is known that only singular controls can be candidates for traps in the quantum control landscape for a controllable quantum system and that all constant controls are singular for $N$-level closed quantum systems with $N\geq3$~\cite{Wu_Long_Dominy_Ho_Rabitz_2012}. At the moment, all known examples of higher-order traps~\cite{deFouquieresSchirmer,PechenTannor2011,PechenTannor2012,VolkovPechen}, like all examples of true traps~\cite{deFouquieresSchirmer} in the dynamical landscapes for a controllable quantum system, are constant controls.

For a controllable $\Lambda$-type quantum system, if the final time $T$ is sufficiently large then minimal and maximal values of the objective functional $J_O$ are reachable and equal to $\lambda_2$ and $\lambda_1$ respectively, where $\lambda_2<J_O(f_0)=\lambda_3<\lambda_1$. In this case, the constant zero control $f_0$ is a singular control and a critical point of the objective functional (see Appendix~\ref{A}), but is not a global extremum point. Therefore, this control can potentially be a trap of some order. 
Below we prove that indeed $f_0$ is a trap and find its order. We directly show the existence of a direction in a vicinity of $f_0$ along which the objective grows. By presenting such a direction, we thereby prove that $f_0$ is not a global maximum point of the objective functional, without using the fact that the system is controllable.
We show that $f_0$ is a third order trap also for some uncontrollable quantum systems.

Our main analytical result is the following theorem.
\begin{theorem} For a sufficiently large time $T$,
\begin{itemize}
    \item If $(H_0,V)$ is an anharmonic system, i.e. $\omega_1\neq \omega_2$, then the zero constant control $f_0\equiv 0$ is a \textbf{trap of the third order} for the objective functional $J_O$. 
    \item If $(H_0,V)$ is a controllable harmonic system , i.e. $\omega_1=\omega_2$ and $|v_{12}|\neq |v_{23}|$, then the zero constant control $f_0\equiv 0$ is a \textbf{trap of the seventh order} for the objective functional $J_O$.      
  \end{itemize}
\end{theorem}

The details of the proof of this theorem are given in Appendices~\ref{A},~\ref{B}, and~\ref{C}. All Fr\'{e}chet variations up to the eight order of the objective functional $J_O$ at the point $f_0\equiv 0$ are computed in Appendix~\ref{A}. A precise mathematical formulation and proofs of the assertion about third order traps for the anharmonic systems and seventh order traps for the harmonic systems are provided in Appendices~\ref{B} and~\ref{C}, respectively.

This theorem shows that symmetry due to degeneracy of the Bohr frequencies in $\Lambda$-type systems leads to a much more severe (mathematically) trapping feature than for such systems without this symmetry and degeneracy. However, how significantly traps of the seventh order slow down the practical optimization compared to traps of the third order, a priori is not evident. In Sec.~\ref{Sec:Numerical}, we numerically analyze and compare these cases.

\section{Control landscape analysis for the uncontrollable $\Xi$ system}
\label{uncontrollable_Xi_system}

The study of the remaining case of the three-level system with one forbidden direct transition, which is the uncontrollable harmonic system (uncontrollable $\Xi$-system), requires a separate analysis. In this section, we analyze for this system by purely analytical methods the problem of local extrema. As a result, we show that in this case traps in the sense of Definition~\ref{DefTraps} may arise. In particular, the zero constant control can be such a trap. 

For technical simplicity, we will consider not the most general case but the case of the same dipole moments for the two allowed transitions. So we assume that $\omega_1=\omega_2=\omega\neq 0$ and $v_{12}=v_{23}$. We also assume, without loss of generality, that $\mathrm{Tr}H_0=0$, i.e. $h_2=0$, $h_1=-\omega$, $h_3=\omega$. 
Then the free Hamiltonian $H_0$ and the interaction Hamiltonian $V$ generate a spin-1 representation of the Lie algebra $\mathfrak{su}(2)$, i.e. $\mathfrak{Lie}(iH_0,iV)\cong \mathfrak{su}(2)$ (for $A,B\in \mathfrak{u}(N)$ by $\mathfrak{Lie}(A,B)$ we denote the minimal Lie algebra containing $A$ and $B$). Let $J_x$, $J_y$ and $J_z$ be standard generators of spin-1
representation of the Lie algebra $\mathfrak{su}(2)$ with matrices:
$$
J_x=\frac 1{\sqrt{2}}
\begin{pmatrix}
0&1&0\\
1&0&1\\
0&1&0
\end{pmatrix},\;
J_y=\frac 1{\sqrt{2}}
\begin{pmatrix}
0&-i&0\\
i&0&-i\\
0&i&0
\end{pmatrix}, \;
J_z=
\begin{pmatrix}
1&0&0\\
0&0&0\\
0&0&-1
\end{pmatrix}.
$$
Then $H_0=-\omega J_z$
and $V=v_xJ_x+v_yJ_y$, where $v_x=\sqrt{2}\operatorname{Re}\,v_{12}$ and $v_y=-\sqrt{2}\operatorname{Im}\,v_{12}$.

Let $\mathcal R$ be a closed connected Lie subgroup of $\operatorname{SU}(3)$ such that its tangent space at the identity coincides with $\mathfrak{Lie}(iH_0,iV)$. Since the Lie group $\mathcal R$ is simple it is a reachable set from the identity for a sufficiently large time $T$ of the system $(H,V_0)$~\cite{JurdjevicSussmann,Jurdjevic,deFouquieresSchirmer}. In other words, there exists $T_0>0$ such that for any time $T\geq T_0$ the group $\mathcal R$ coincides with the set of all operators $U\in \operatorname{SU}(3)$ such that $U=U_T^f$ for some $f\in L_2([0,T],\mathbb{R})$.

Consider the function $F_O\colon \operatorname{SU}(3) \to \mathbb{R}$ defined as
$F_O(U)=\mathrm{Tr}(OU\rho_0U^\ast)$. We call this function the kinematic functional of the quantum problem~\cite{PechenTannor2011}. We have that $J_O(f)=F_O(U_T^f)$. It is known that this function does not have so-called kinematic traps, i.e. there are no points of local but not global extrema of the function $F_O$ on the Lie group $\operatorname{SU}(3)$~\cite{RHR}. Let $F_1=\left.F_O\right|_\mathcal{R}$ be the restriction of the kinematic functional $F_O$ on the Lie group $\mathcal{R}$. Below we show that the function $F_1$ has points of local but not global maxima on the Lie group $\mathcal{R}$.

Let $\mathcal{R}_Z$ be a Lie subgroup $\mathcal{R}_Z=\{U\in \mathcal{R}\colon U=e^{-i\phi J_Z},\phi\in [-\pi,\pi)\}\cong \operatorname{U}(1)$ of the Lie group $\mathcal R$. Since $H_0=-\omega J_z$, this set of operators coincides with those operators that can be obtained using only the free evolution.

\begin{proposition}
\label{MimMaxF1}
The global minimum and maximum of the function $F_1$ are
\begin{align*}
\max_{U\in \mathcal{R}}F_1(U)&=\lambda_1;\\
\min_{U\in \mathcal{R}}F_1(U)&=\frac{\lambda^2_2-\lambda_1\lambda_3}{2\lambda_2-\lambda_1-\lambda_3}<\lambda_3.
\end{align*}
An operator $U\in \mathcal{R}$ is a local but not global extremum point of the function $F_1$ if and only if $U\in\mathcal{R}_Z$. These are points of local maximum, the value at which is $\lambda_3$.
\end{proposition}

The following theorem follows from this proposition.

\begin{theorem}
\label{Thmuncontrollablecase}
Let $(H_0,V)$ be an uncontrollable harmonic $\Xi$ system with $v_{12}=v_{23}$ and $\mathrm{Tr}H_0=0$. For $T\geq T_0$ the control $f\in L_2([0,T],\mathbb{R})$ is a true trap for the functional $J_O$ if and only if 
$U_T^f\in \mathcal{R}_Z$. In particular, $f_0\equiv 0$ is a true trap.
\end{theorem}
Proofs of Proposition~\ref{MimMaxF1} and Theorem~\ref{Thmuncontrollablecase} are given in Appendix~\ref{D}.

\begin{remark}
Note that due to uncontrollability of the quantum system $(H_0,V)$ the minimum value of the objective functional is greater than for the case of a controllable quantum system:
$$
\min_{f\in L_2([0,T],\mathbb{R})} J_O(f)=\frac{\lambda^2_2-\lambda_1\lambda_3}{2\lambda_2-\lambda_1-\lambda_3}>\lambda_2,
$$
but the maximum values of the objective functional for the uncontrollable and the controllable systems are the same.
\end{remark}

\section{Influence of the higher order traps on the efficiency of the optimization}\label{Sec:Numerical}

In this section we investigate, using the GRAPE local search algorithm~\cite{GRAPE_2}, how significantly traps of the seventh order slow down the practical optimization compared to traps of the third order for the three-level $\Lambda$-type quantum systems with Hamiltonians of the form~(2). We consider two particular examples, system $S_1$ and system $S_2$, with the free and interaction Hamiltonians, respectively, of the form:
\begin{align*}
S_1:&\quad H_0= 
\begin{pmatrix}
0 & 0 & 0\\
0&  1 & 0\\
0& 0 & 2.5
\end{pmatrix},\quad 
V=\begin{pmatrix}
0 & 1 & 0\\
1 & 0 & 1.7\\
0 & 1.7 & 0
\end{pmatrix}\\
S_2:&\quad H_0= 
\begin{pmatrix}
0 & 0 & 0\\
0 & 1 & 0\\
0 & 0 & 2
\end{pmatrix},\quad 
V=\begin{pmatrix}
0 & 1 & 0\\
1 & 0 & 1.7\\
0 & 1.7 & 0
\end{pmatrix}
\end{align*}
The system $S_1$ is a nondegenerate $\Xi$-system~[Fig.~\ref{Fig1:structure}(d)], while the system $S_2$ is a controllable degenerate $\Xi$-system~[Fig.~\ref{Fig1:structure}(f)]. As a target quantum observable we consider
\begin{eqnarray*}
O=\sum_{j=1}^3{\lambda_j|j\rangle\langle j|} = |1\rangle\langle 1| - \lambda|2\rangle\langle 2|,
\end{eqnarray*} 
where $\lambda=5$, and as the initial density matrix we take $\rho = |3\rangle\langle 3|$. Then for sufficiently large final times $T$ at which the systems are controllable, $\max J_O=1$ and $\min J_O=-5$.

\begin{figure}[t]
\center{\includegraphics[width=.7\columnwidth]{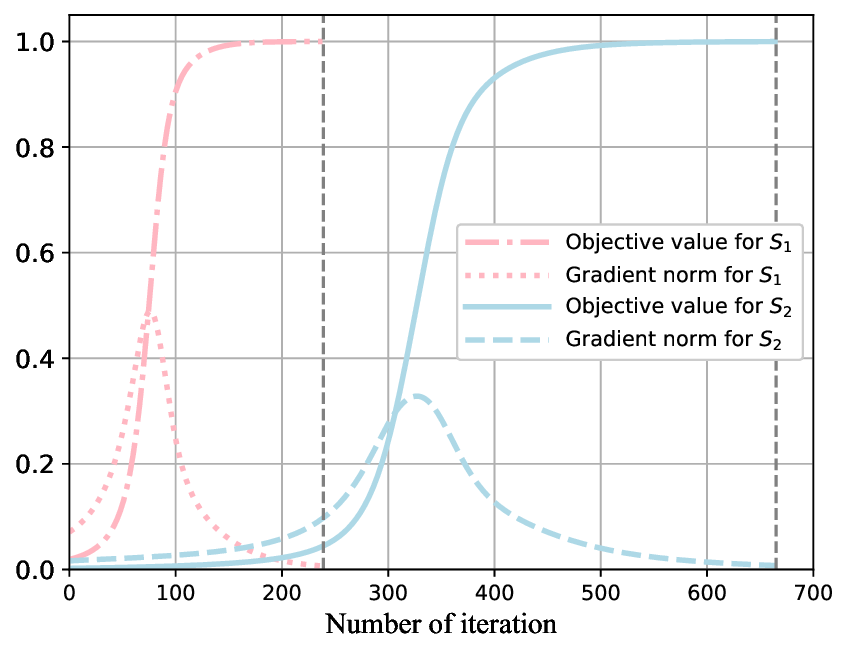}}
\caption{A typical dependence of the gradient norm and of the objective value on the number of iterations for the systems $S_1$ (pink lines) and $S_2$ (blue lines) for $\lambda = 0$ for GRAPE starting in a hypercube of size $l$ around zero control. The parameters are $T=10$, $D=200$, $I_{\mathrm{err}} = 10^{-4}$, $\varepsilon = 0.1$ and $l=2$ for system $S_1$, $l=4.5$ for system $S_2$. The dashed lines show the iteration at which the algorithm was completed.}
\label{fig:S_1_grad}
\end{figure}

\begin{figure}[t]
\centering\includegraphics[width=.7\linewidth]{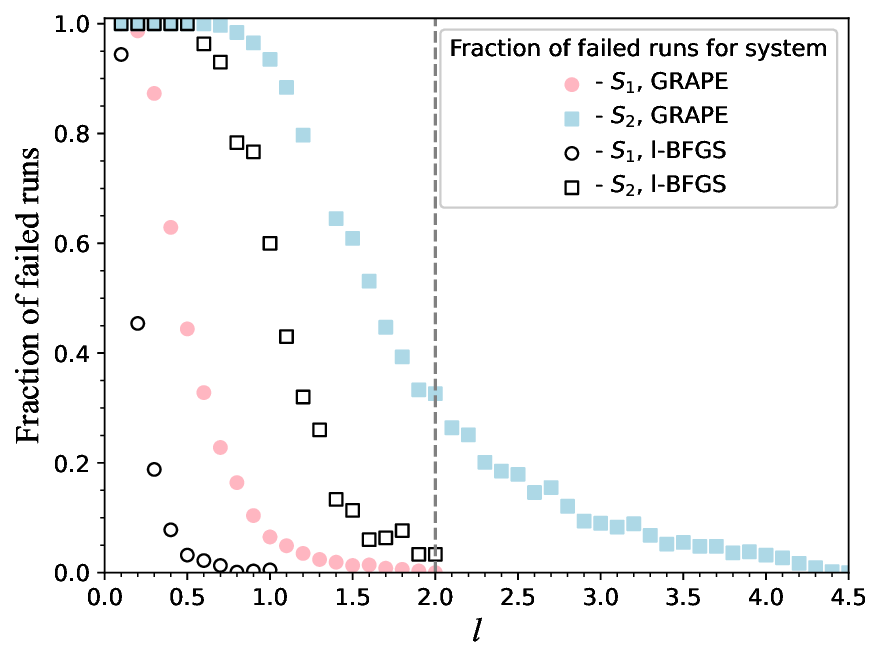}
\caption{The dependence of the fraction of failed runs of the GRAPE and l-BFGS algorithms on the parameter $l$ for the anharmonic system $S_1$ with the third order trap (pink and empty circles for GRAPE and l-BFGS, respectively) and for the harmonic system $S_2$ with the seventh order trap (blue and empty squares for GRAPE and l-BFGS, respectively). The parameters are $T=10$, $D=200$, and {$I_{\mathrm{err}}=10^{-4}$}.  The initial controls are generated randomly in the interval $[-l, l]$ and $L=10^3$ runs are performed to generate each point. For GRAPE, fixed step size is $\varepsilon=0.1$ and stopping criterion is defined by $K_{\mathrm{stop}}=2000$. For l-BFGS, the stopping criterion is defined by the gradient norm less or equal than $10^{-5}$ and $K_{\mathrm{stop}}=100$. }
\label{failed_runs}
\end{figure}

\begin{figure*}
\center{\includegraphics[width=\linewidth]{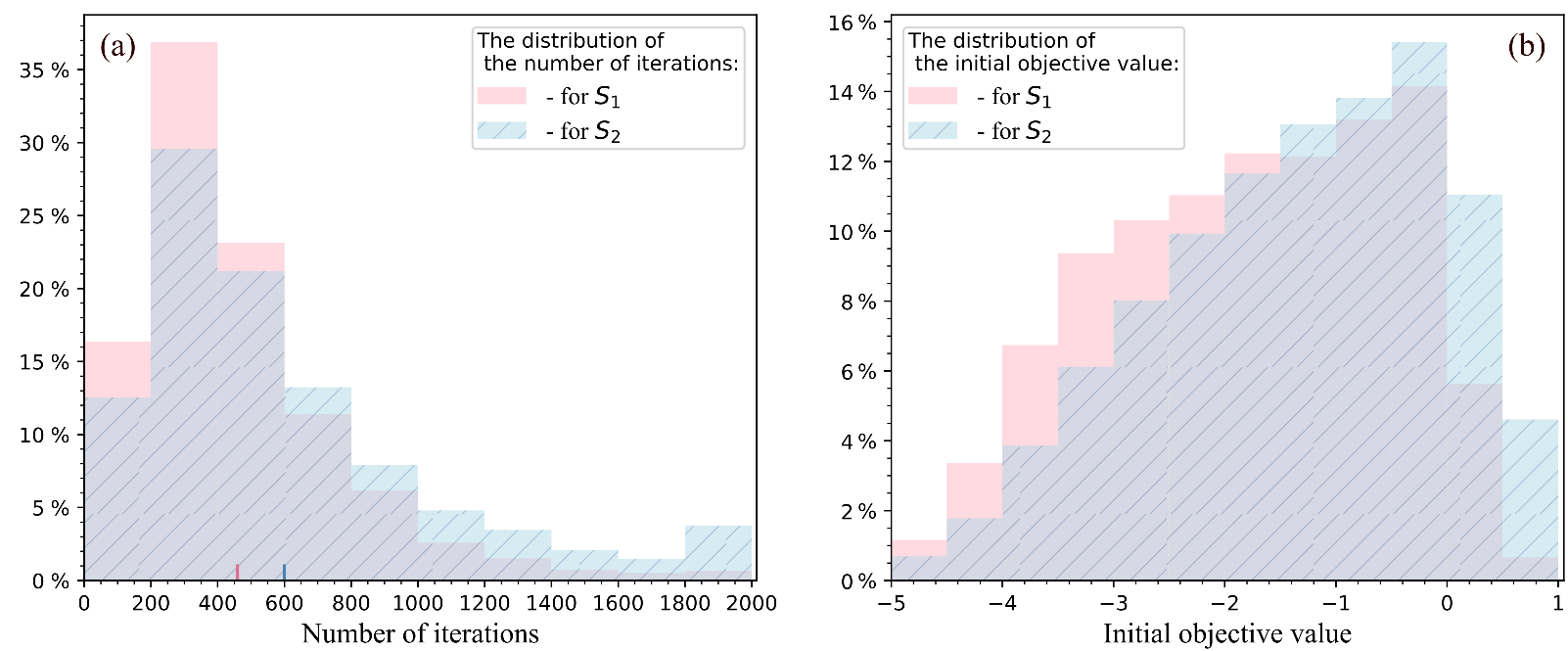}}
\caption{Histogram (a) represents the distribution of the number of iterations required to reach the objective value $J_{\mathrm{stop}}$ when $N_{\rm fail}=0$ starting from a large enough interval around zero control. Mean values $459$ and $599$ are shown by vertical ticks. Histogram (b) represents the distribution of the initial objective values. Totally 10 intervals are used for left histogram with width of interval 200 and 12 intervals used for right histogram with width of interval 0.5. The parameters are $T=10$, $D=200$, $I_{\mathrm{err}}=10^{-4}$, $\varepsilon=0.1$, $L=10^4$; $l=2$ for $S_1$ (pink plane) and $l=4.5$ for $S_2$ (blue shaded). The ordinate shows the percentage of the total number of runs $L$.}
\label{fig:3}
\end{figure*}

\begin{figure*}[!h]
\center{\includegraphics[width=\linewidth]{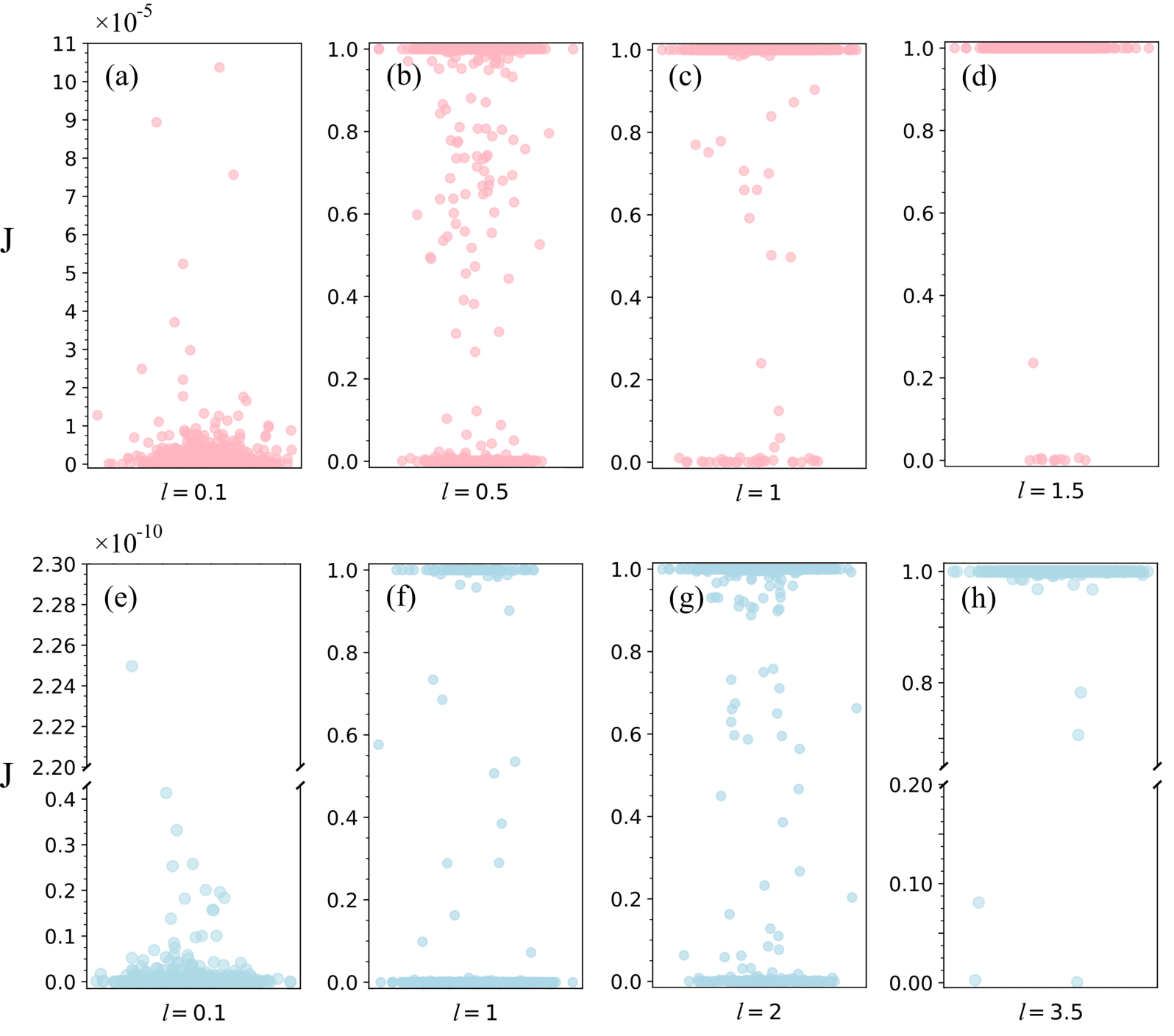}}
\caption{The values of the objective functional $\mathcal{J}_O$ at which the GRAPE algorithm stops when the initial controls are generated randomly in the hypercube $[-l, l]^D$ around zero control. The first row: for the anharmonic system $S_1$: (a) $l = 0.1$; (b) $l=0.5$; (c) $l = 1$; (d) $l=1.5$. The second row: for the harmonic system $S_2$ with parameters: (e) $l = 0.1$; (f) $l=1$; (g) $l = 2$; (h) $l=3.5$. The parameters are $T=10$, $D=200$, $I_{\mathrm{err}}=10^{-4}$, $\varepsilon=0.1$, $L=10^3$. The abscissa is used just for convenience to show all the points.}
\label{fig:scatters_for_S}
\end{figure*}
 
\begin{figure}[!h]
\centering\includegraphics[width=.7\linewidth]{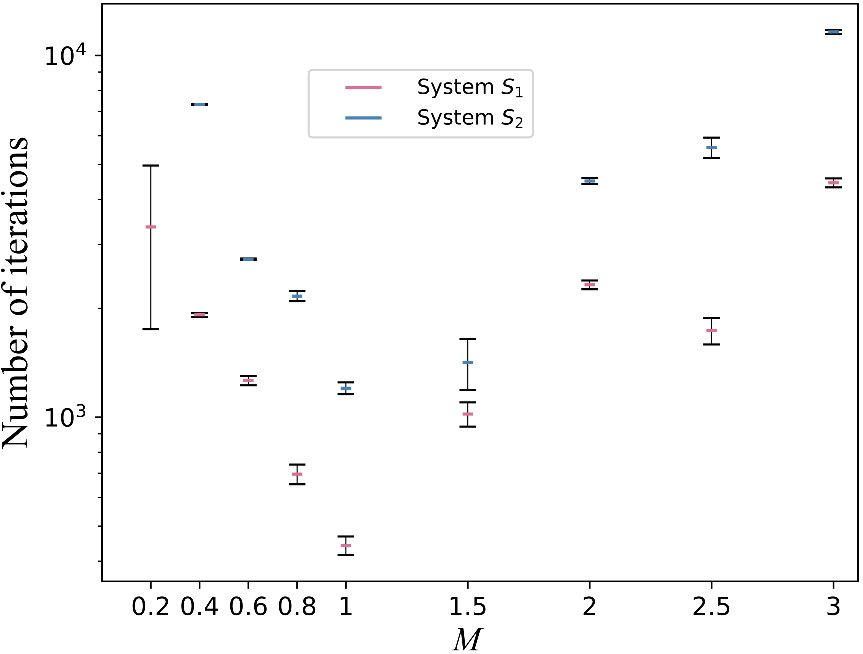}
\caption{The average number of iterations required for successful runs of the GRAPE algorithm with initial controls generated randomly in $[-M-0.1, -M+0.1]^D$ depending on the offset $M$ for the anharmonic system $S_1$ (pink dashes) and for the harmonic system $S_2$ (blue dashes). For each $M$, the average number of iterations for $S_2$ is larger than for $S_1$. Error bars show standard deviations. For $S_2$ and $M=0.2$, $K_{\mathrm{stop}}=14\times 10^3$ iterations were not enough and hence the corresponding point is not shown. The parameters are $T=10$, $D=200$, $I_{\mathrm{err}}=10^{-4}$, $\varepsilon=0.1$, $K_{\mathrm{stop}}=14\times10^{3}$. Each point is the average over runs of GRAPE starting at $L=10^3$ random initial controls. Vertical scale is logarithmic.} 
\label{fig:different_M}
\end{figure}

\begin{figure}[!h]
\center{\includegraphics[width=\linewidth]{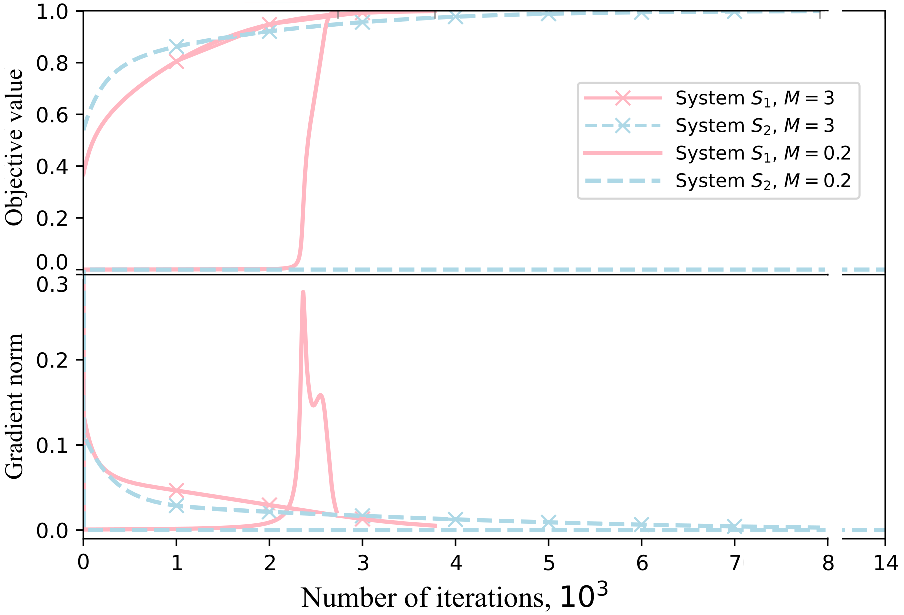}}
\caption{A typical dependence of the objective value (upper plot) and gradient norm (bottom plot) on the number of iterations for the system $S_1$ (pink solid lines) and for the system $S_2$ (blue dashed lines) for $l=0.1$ and $M=0.2, 3$. The parameters are $T=10$, $D=200$, $I_{\mathrm{err}} = 10^{-4}$, $\varepsilon = 0.1$. The strokes at the top of the figure show the iteration numbers at which the algorithm stopped.}
\label{Fig7}
\end{figure}

The GRAPE algorithm for these systems has been implemented manually using the NumPy library from Python. For each system, we consider piecewise constant controls of the form 
\begin{equation*}
f_C(t)=\sum_{k=1}^D c_k\chi_{[t_k,t_{k+1}]}(t),
\end{equation*}
where $C = (c_1,\dots,c_D)$ is a $D$-dimensional vector, $\chi_{[t_k,t_{k+1}]}(t)$ is the characteristic function of the interval $[t_k,t_{k+1}]$, $t_k = \Delta t(k-1)$, and $\Delta t = T/D$. The control goal is to maximize the objective function $\mathcal{J}_O\colon \mathbb{R}^D\to \mathbb{R}$ determining the average value of the observable $O$:
\begin{equation*}
\mathcal{J}_O(C) = J_O(f_C)=\Tr[U_T^{f_C} \rho_0(U_{T}^{f_C})^{\dagger}O].
\end{equation*}
Gradient of the objective function has the form:
\begin{eqnarray*}
\nabla \mathcal{J}_O(C)= \left(\frac{\partial \mathcal{J}_O(C)}{\partial c_1},\ldots,\frac{\partial \mathcal{J}_O(C)}{\partial c_D}\right).
\end{eqnarray*}
Here the partial derivative with respect to $c_k$ in the linear in $\Delta t$ approximation is
\begin{equation}\label{gradient}
\frac{\partial \mathcal{J}_O(C)}{\partial c_k} \approx 2\Delta t \times \operatorname{Im}(\Tr[W_k^{\dagger}VW_k\rho_0W_D^{\dagger}OW_D]),
\end{equation}
where $W_k = U_kU_{k-1}\dots U_2U_1$, $U_n = e^{-i(H_0+c_nV)\Delta t}$, and $U^{f_C}_T=W_D$.

Let us describe the operation of the GRAPE algorithm. The algorithm starts with constructing an initial control $C_1$, whose components are generated randomly with a uniform distribution in the interval $[-l, l]$ with some $l>0$. On $i$-th iteration, starting from control $C_i$ we compute gradient $\nabla \mathcal {J}_O(C_i)$ and update the control as $C_{i+1}=C_{i}+\varepsilon\cdot\nabla \mathcal{J}_O(C_i)$, where $\varepsilon>0$ is a fixed small number. The loop continues either until $\mathcal{J}_O$ reaches the fidelity value $J_{\mathrm{stop}}=1-I_{\mathrm{err}}$, where $I_{\mathrm{err}}$ is some threshold describing admissible deviation from the global maximum objective, or until a predefined maximum admissible number of iterations $K_{\mathrm{stop}}$ is reached. For simulations we use final time $T=10$, number of control vector components $D=200$, random initial control vectors in the hypercube $[-l,l]^D$, maximum admissible number of iterations $K_{\mathrm{stop}} = 2000$, step size $\varepsilon = 0.1$, and admissible deviation from the global maximum objective $I_{\mathrm{err}} = 10^{-4}$. For this $I_{\mathrm{err}}$ we have $J_{\mathrm{stop}}=0.9999$. If starting from some initial control the GRAPE algorithm is not able to reach the value $J_{\mathrm{stop}}$ with not more than $K_{\mathrm{stop}}$ iterations, we call this run a failed run.

Fig.~\ref{fig:S_1_grad} shows the dependence of the gradient norm and of the objective value on the number of iterations for the systems $S_1$ and $S_2$ when $\lambda=0$. In this case, for both systems it is easy to obtain high values of the fidelity. The dependence of the gradient norm and objective value on the number of iterations in this case is similar for different runs and is also similar to that obtained earlier for the $\Lambda$-atom~(Fig.~2 in~\cite{PechenTannor2012}) for $\lambda=0$. 

To estimate how the order of the trap affects the efficiency of the GRAPE algorithm, now we consider $\lambda=-5$, when zero control becomes a third order trap for the system $S_1$ and a seventh order trap for the system $S_2$. For this, for each $l=0.1,0.2\dots,2$ for the system $S_1$ and $l=0.1,0.2,\dots,4.5$ for the system $S_2$ we compute the number of failed runs $N_\mathrm{fail}$ (among runs for $L=10^3$ initial controls randomly generated in the hypercube $[-l, l]^D$) at which the GRAPE algorithm has not been able to obtain the required accuracy of the objective functional value. Fig.~\ref{failed_runs} shows the dependence of the fraction of failed runs $N_\mathrm{fail}/L$ on $l$ for the anharmonic system $S_1$ with the third order trap (pink circles) and for the harmonic system $S_2$ with the seventh order trap (blue squares).

We find that the order of the trap significantly affects the efficiency of the GRAPE algorithm. For the harmonic system $S_2$ having the seventh order trap, in difference to the anharmonic system $S_1$ having the third order trap, it was necessary to significantly increase the size of the interval $l$ in which the components of the initial control vector are generated, in order to always achieve the desired maximal fidelity value $J_{\rm stop}$ of the objective functional. The required fidelity of the objective functional for the anharmonic system $S_1$ (pink circles) is obtained for all $L=10^3$ runs (i.e., when the number of failed runs $N_{\mathrm{fail}}=0$) for $l=2$ (for $l=1$ the fraction of failed runs is as small as about 5\%), while for obtaining the same for the anharmonic system $S_2$ (blue squares) the parameter $l$ should be increased to $l=4.5$ (and 5\% fraction of failed runs is at about $l=3.5$). Similarly, we find that the mean number of iterations required to reach the desired objective value starts to decrease for the harmonic system $S_1$ from $l=0.3$, while for the anharmonic system $S_2$ it starts to decrease only from $l=0.8$ and with a slower speed.

For the numerical analysis of the control landscapes we use a first-order GRAPE algorithm with a fixed step size to investigate and compare different landscapes. In practice, GRAPE versions based on higher-order quasi-Newton methods, such as Broyden–Fletcher–Goldfarb–Shanno (BFGS) algorithm, including its limited-memory version l-BFGS, are often used in quantum control~\cite{EitanPRA2011} (as well as in other optimization problems) and are known to be more efficient. For a comparison, we investigate and show on Fig.~\ref{failed_runs} the influence of the 3rd and the 7th order traps on the optimization by l-BFGS. We use Python implementation of l-BFGS realized in the {\tt minimize} function of the {\tt SciPy} library. The objective function and the gradient with approximate derivatives~(\ref{gradient}) are passed as the parameters to the {\tt minimize} function. A run is considered as successful if the final obtained value of the objective function is more than ${1-I_{\mathrm{err}}}$. For l-BFGS, the stopping criterion is defined by the gradient norm less or equal than $10^{-5}$ or by the number of iterations more than $K_{\rm stop}=100$ (this significantly smaller maximal number of iterations is used since generally l-BFGS achieves the objective in much less iterations than the first-order fixed step size GRAPE). In overall, as expected we find that l-BFGS works much better than fixed step size GRAPE and in all successful runs it ends in less than 100 iterations. Regarding the traps we find, as shown on Fig.~\ref{failed_runs}, that while the number of steps necessary to obtain the desired value of the fidelity by l-BFGS decreases roughly by the order compared to GRAPE, their trapping behavior remains similar, i.e. to optimize in a vicinity of the 7th order trap by l-BFGS
is significantly more difficult than to optimize in the vicinity of the 3rd order trap. In the rest of the work, we use the fixed step size GRAPE for further investigation of the control landscapes.

The distribution of the number of iterations required to achieve objective value $J_{\mathrm{stop}}$ and the distribution of the initial objective values for the anharmonic system $S_1$ with $l=2$ and for the harmonic system $S_2$ with $l=4.5$, i.e. when the number of failed runs becomes zero, are shown on Fig.~\ref{fig:3} (10 bins with width of 200 and 12 bins with width 0.5 were taken for (a) and (b) subfigures, respectively). On average $459$ and $599$ iterations (shown by vertical ticks) are required to get $J_{\mathrm{stop}}$ population transfer using GRAPE for $S_1$ and $S_2$, respectively. 
Here $L=10^4$ random initial controls were used to generate the plots for each system. The left histogram shows that the distributions of the number of iterations for both systems are very similar when initial controls are generated randomly in the hypercubes with the corresponding sizes $l$ for which the algorithm always reaches the given fidelity. Moreover, these distributions are relatively narrow and concentrated around mean values. So, obtaining numbers of iterations significantly larger than the mean value are rare. The distribution of the initial objective values, in opposite, is quite flat and has no narrow peaks. Even low initial objective values are possible with not too small probability. From both diagrams it can be seen that even though the algorithm runs from various, including small, objective values, starting from a sufficient distance from the trap allows to always reach the fidelity value $J_{\rm stop}$ in the admissible number of iterations.

Fig.~\ref{fig:scatters_for_S} shows the values of the objective function $\mathcal{J}_O$ at which the GRAPE algorithm stops for the systems $S_1$ and $S_2$. Here for each $l$ totally $L=10^3$ initial controls generated randomly in the interval $[-l, l]$ were used to produce the final objective values using the GRAPE algorithm. The top row shows the results for the anharmonic system $S_1$ with $l=0.1,0.5,1,1.5$, the bottom row shows the results for the harmonic system $S_2$ with $l=0.1,1,2,3.5$. We can see that even when choosing a substantially large interval for the initial control vector, there are cases in which we get stuck in the neighborhood of the trap and cannot approach the global maximum.

Next we compare the obtained results for the same systems $S_1$ and $S_2$ but for initial control vectors $C$ whose components are generated randomly in the interval $[-l - M, l - M]$ with constant offset $M$. First, for each system we make $L=10^3$ runs of the algorithm without limiting the number of iterations $K_{\mathrm{stop}}$ ($K_{\mathrm{stop}}=\infty$) for $l=0.1$ and $M=0$. It was found that to reach the objective value $J_{\mathrm{stop}}$ for an anharmonic system could require about $1.2\times10^{5}$ iterations, and for a harmonic system more than $5\times10^{5}$ iterations.
Next for calculations, we take $M=0.2,0.4,0.6,0.8,1,1.5,2,2.5,3$ and a large value $K_{\mathrm{stop}}=14\times10^3$. Fig.~\ref{fig:different_M} shows the average number of iterations required to reach the objective value $J_{\mathrm{stop}}$ when $l=0.1$ for the systems $S_1$ (pink dashes) and $S_2$ (blue dashes) with $M$ offset together with obtained standard deviations. For $S_2$ and $M=0.2$, $K_{\mathrm{stop}}=14\times 10^3$ iterations were not enough (hence the corresponding point is not shown). For each $M$ the average number of iterations for $S_2$ is larger than for $S_1$. The computed numbers of iterations significantly decrease with $M$ up to $M=1$. However, for $M>1$ they again start to increase. To reveal the reason for this increase, we computed the number of iterations also for $M=3.5, 4, 4.5, 5, 5.5$ (not shown) and did not find any maximum at some $M$ that could indicate another trap within this range. To see the difference between $M=0.2$ and $M=3$ cases, when the numbers of iterations are similar, on Fig.~\ref{Fig7} we show a typical behavior of the objective (upper plot) and gradient norm (bottom plot) vs number of iterations for both systems for these values of $M$. When $M=0.2$, for the system $S_1$ gradient norm is very small at the beginning (close to the trap) but is sufficient to move away of the trap, and after about $2\times 10^3$ iterations gradient norm sharply increases forcing the objective in some hundreds of iterations to achieve a close to $1$ value. For the system $S_2$, gradient norm is very small during all the iterations so the objective almost does not increase (objective does not exceed $3.9\times 10^{-9}$) during all the $14\times 10^3$ iterations. For $M=3$, for both systems gradient remains small and decreasing so a large number of iterations is required to reach the target objective value even in the absence of trap. Thus as we can see, the addition of a not very large offset significantly affects the capabilities of avoiding the trap and allows to significantly reduce the running time of the algorithm and make the optimization more efficient.

\section{Conclusions}\label{Sec:Conclusions}
In~\cite{PechenTannor2011}, the notion of higher order traps for quantum control landscapes was introduced and $\Lambda$-atom with non-degenerate energy states was shown to be the simplest quantum system having a second (and actually a third) order trap. The case of degenerate transition Bohr frequencies was not treated, while the conventional $\Lambda$-atom is an important example of such system. In this work, we have considered systems with degenerate transition frequencies in details. The considered general class of systems also includes $V$-atom and spin-one $\Xi$ ladder-type three-level quantum systems. We rigorously prove that symmetry due to degeneracy of the transition frequencies of a three-level quantum system with one forbidden transition results in the zero constant control being a seventh order trap, while for degeneracy of the energy states this zero constant control remains a third order trap. Thus, zero constant control is a third order trap for a three-level system without degeneracy of the transition frequencies. Hence we find that degeneracy of the transition frequencies leads to a sharp increase in the order of the trap from third to seventh. Our numerical analysis performed for some example cases using GRAPE approach shows that seventh order traps significantly slow down the practical optimization --- the attraction domain of a seventh order trap is much larger than that of a third order trap. In addition to the fixed step size GRAPE, we study the landscapes using l-BFGS and obtain that while l-BFGS is  significantly better than the fixed step size GRAPE, it has a similar trapping behavior in the vicinities of the 3rd and 7th order traps. Important to note that higher order traps are not necessarily traps. At least, no true traps among higher order traps were found up to now, and for some examples, like $\Lambda$-atom, the third order trap was shown to be not a true trap~\cite{PechenTannor2012}. We rigorously prove that found in this work seventh order trap for the degenerate $\Xi$-system is not a true trap either. Thus local search algorithms can escape from it, while at a much slower speed than from a regular point. In addition to the case of a controllable $\Xi$-system, i.e. when the dipole moment matrix elements satisfy~$|v_{12}|\ne |v_{23}|$, we also analytically investigate local extrema for the uncontrollable case which has been studied before in various context~\cite{CookShore,Hioe,TuriniciRabitz,SugnyKontz,Shuang_et_al,ElovenkovaQuantumReport2023} beyond the control landscape analysis. For this uncontrollable system, we find that traps do exist. The trap-free nature of quantum control landscapes was assumed for a fully controllable case~\cite{RHR}, and this finding does not contradict this assumption.

\appendix
\section{Taylor expansion of the objective functional up to the eighth order}
\label{A}

First of all, note that without loss of generality we can assume that $\lambda_1>0$, $\lambda_2<0$, $\lambda_3=0$. Indeed, consider the observable $O'=O-\lambda_3\mathbb I$. Because the functionals $J_O$ and $J_{O'}$ differ by a constant [$J_O(f)=J_{O'}(f)+\lambda_3$], they have equivalent control landscapes~\cite{PechenTannor2011}.

Denote $V_t:=e^{itH_0}Ve^{-itH_0}$. Let $f_0\equiv 0$. The Taylor (Dyson's) series for the evolution operator $U_T^{f_0+\delta f}=U_T^{\delta f}$ for any $\delta f\in L_2([0,T],\mathbb{R})$ has the following form
\begin{eqnarray*}
U_T^{f_0+\delta f} &=&e^{-iTH_0}\Bigl(\mathbb I+\sum_{n=1}^\infty (-i)^n\int_0^Tdt_1\int_0^{t_1}dt_2 \ldots\int_0^{t_{n-1}}dt_n\delta f(t_1)\ldots \delta f(t_n)V_{t_1}\ldots V_{t_n}\Bigr).
\end{eqnarray*}
Substituting this expression in $J_O$, one can obtain the Taylor series for the objective functional $J_O$, $J_O(f_0+\delta f)=\sum\limits_{n=0}^\infty \frac 1{n!} J^{(n)}_O(f_0)(\delta f,\ldots,\delta f)$. Let $A^n_{lk}\colon L_2([0,T],\mathbb{R})\to\mathbb{C}$ be the form of the order $n$ defined as
\begin{eqnarray*}
A^n_{lk}\langle f\rangle & := & \int_0^Tdt_1\int_0^{t_1}dt_2\ldots \int_0^{t_{n-1}}dt_nf(t_1)\ldots f(t_n)\langle l|V_{t_1}\ldots V_{t_n}|k\rangle.
\end{eqnarray*}
By direct calculations, one can obtain the following formula for the Fr\'echet differential of the order $n$ of the objective functional $J_O$ at~$f_0\equiv 0$:
\begin{eqnarray}
\label{Jm1}
\frac 1{n!}J^{(n)}_O(f_0)(f,\ldots,f)
 &=& i^{n} \lambda_1\sum_{j=1}^{n-1} (-1)^{j} A^j_{13}\langle f \rangle \overline{A^{n-j}_{13}\langle f \rangle}  \nonumber \\ &&\quad + i^{n} \lambda_2  \sum_{j=1}^{n-1} (-1)^{j} A^j_{23}\langle f \rangle \overline{A^{n-j}_{23}\langle f \rangle}.
\end{eqnarray}
For the $\Lambda$-type system~(\ref{Lambda-atom}), we have $\langle 2|V_{t_n}\ldots V_{t_1}|3\rangle=0$ for even $n$ and
$\langle 1|V_{t_n}\ldots V_{t_1}|3\rangle=0$ for odd $n$. Hence, if $n$ is even, then $A^{n}_{23}=0$, and if $n$ is odd, then $A^{n}_{13}=0$. Then formula~(\ref{Jm1}) implies that $J^{(n)}_O(f_0)(f,\ldots,f)=0$ if $n$ is odd. Thus $f_0$ is a critical point of the objective functional $J_O$.

Even variations of the objective functional $J_O$ from the second to the eighth order have the form
\begin{eqnarray}
\label{J2}
\frac 1{2!}J_O^{(2)}(f_0)(f,f)&=&\lambda_2 |A^{1}_{23}\langle f \rangle|^2,\nonumber \\
\label{J4}
\frac 1{4!}J_O^{(4)}(f_0)(f,\ldots, f)&=& \lambda_1|A^{2}_{13}\langle f\rangle|^2
-\lambda_2\left(A^1_{23}\langle f\rangle\overline{A^{3}_{23}\langle f \rangle}+A^3_{23}\langle f\rangle\overline{A^{1}_{23}\langle f\rangle}\right),\\
\label{J6}
\frac 1{6!}J_O^{(6)}(f_0)(f,\ldots, f)&=& \lambda_2\left(|A^{3}_{23}\langle f\rangle|^2+A^1_{23}\langle f\rangle\overline{A^{5}_{23}\langle f \rangle}+A^5_{23}\langle f\rangle\overline{A^{1}_{23}\langle f\rangle}\right)-\nonumber\\&&-\lambda_1\left(A^2_{13}\langle f\rangle\overline{A^{4}_{13}\langle f \rangle}+A^4_{13}\langle f\rangle\overline{A^{2}_{13}\langle f\rangle}\right),\\
\label{J8}
 \frac 1{8!}J_O^{(8)}(f_0)(f,\ldots, f)  &=&\lambda_1\left(|A^{4}_{13}\langle f\rangle|^2+A^2_{13}\langle f\rangle\overline{A^{6}_{13}\langle f \rangle}+A^6_{13}\langle f\rangle\overline{A^{2}_{13}\langle f\rangle}\right)-\nonumber\\&&-\lambda_2(A^1_{23}\langle f\rangle\overline{A^{7}_{23}\langle f \rangle}+A^7_{23}\langle f\rangle\overline{A^{1}_{23}\langle f\rangle}+A^3_{23}\langle f\rangle\overline{A^{5}_{23}\langle f \rangle}+\nonumber\\&&+A^5_{23}\langle f\rangle\overline{A^{3}_{23}\langle f\rangle}).
\end{eqnarray}

\section{Anharmonic system. Trap of the third order}
\label{B}

Let $\mathfrak{H}^0=L_2([0,T],\mathbb{R})$.
Define the linear space of controls $\mathfrak{H}^1\subset\mathfrak{H}^0$ as
\begin{equation*}
\mathfrak{H}^{1}=\{f\in \mathfrak{H}^0\colon A^{1}_{23}\langle f \rangle=0\}.
\end{equation*}
Since $A^{1}_{23}\langle f \rangle=v_{23}\int_0^Te^{-i\omega_2t}f(t)dt$, the space $\mathfrak{H}^1$ consists of controls $f\in \mathfrak{H}^0$ which satisfy
$\int_0^Te^{-i\omega_2t}f(t)dt=0$.

We have the following result, originally obtained in~\cite{PechenTannor2011,PechenTannor2012} for the case of the $\Lambda$-atom (with $0<-\omega_2<\omega_1$ in our notations).
\begin{theorem}
\label{thm1}
The Taylor expansion for the objective functional $J_O$ at the constant zero control $f_0\equiv 0$ has the form
\begin{multline*}
J_O(f_0+\delta f)=J_O(f_0)+\frac 12J_O^{(2)}(f_0)(\delta f,\delta f)\\+\frac 12J_O^{(4)}(f_0)(\delta f,\ldots,\delta f)+o(\|\delta f\|^{5})\text{\;as $\|\delta f\|\rightarrow 0$,}
\end{multline*}
where for all $\delta f\in \mathfrak{H}^0\setminus \mathfrak{H}^1$ one has $J_O^{(2)}(f_0)(\delta f,\delta f)<0$ and for all $\delta f\in \mathfrak{H}^1$ one has $J_O^{(2)}(f_0)(\delta f,\delta f)=0$ and
$J_O^{(4)}(f_0)(\delta f,\ldots,\delta f)
\geq0$. Let $(H_0,V)$ be an anharmonic system, i.e. $\omega_1\neq \omega_2$. Then
\begin{enumerate}
\item if $\omega_2=0$ then for any $T> 0$ there exists control $\delta f\in \mathfrak{H}^1$ such that $J_O^{(4)}(f_0)(\delta f,\ldots,\delta f)>0$;
\item if $|\omega_2|>0$ then for any $T\geq T_0:=\frac{2\pi}{|\omega_2|}$ there exists control $\delta f\in \mathfrak{H}^1$ such that $J_O^{(4)}(f_0)(\delta f,\ldots,\delta f)>0$.
\end{enumerate}
Thus, $f_0$ is a trap of the third order for the objective functional $J_O$.
\end{theorem}
\begin{proof}
Formula~(\ref{J2}) implies that for all $f\in \mathfrak{H}^0\setminus\mathfrak{H}^1$ (for shortness of notation, in the proof we use $f$ instead of $\delta f$) we have $J^{(2)}_O(f_0)(f, f)<0$
and for all $f\in \mathfrak{H}^1$ we have
\begin{equation*}
J_O^{(2)}(f_0)(f,f)=0.
\end{equation*}
Since $A^{1}_{23}\langle f \rangle=0$ for all $f\in \mathfrak{H}^1$, formula~(\ref{J4}) implies that for $f\in \mathfrak{H}^1$ holds
\begin{equation*}
\frac 1{4!}J_O^{(4)}(f_0)(f,\ldots, f)=\lambda_1|A^{2}_{13}\langle f\rangle|^2\geq 0,
\end{equation*}
where 
\begin{equation*}
A^{2}_{13}\langle f\rangle
=v_{12}v_{23}\int_0^Tdt_1\int_0^{t_1}dt_2e^{-i\omega_1t_1-i\omega_2t_2} f(t_1)f(t_2).
\end{equation*}

Now we show that if $\omega_1\neq\omega_2$ then the set $\mathfrak{H}^{1}$ contains a nonzero control $f\ne0$ such that the value of the Fr\'{e}chet differential of the 4-th order on this control is positive, $J_O^{(4)}(f_0)(f,\ldots, f)>0$. For this purpose we show that there exists a nonzero control $f\in\mathfrak{H}^{1}$ such that $A^{2}_{13}\langle f\rangle\neq 0$.

Let $\omega_2=0$ (and $\omega_1\ne0$) [Fig.~\ref{Fig1:structure}(b)]. Consider the control $f_1(t)=\chi_{[0,T]}(t)$, where $\chi_{[0,T]}$ is the characteristic function of the segment $[0,T]$.
Then
\begin{eqnarray}
A^{2}_{13}\langle f_1\rangle
=v_{12}v_{23}\int_0^{T}dt_1\int_0^{t_1}dt_2e^{-i\omega_1t_1}=v_{12}v_{23}\frac{-1+e^{-i\omega_1T}(1+i\omega_1T)}{\omega_1^2}\neq 0.\nonumber
\end{eqnarray}

Let $\omega_2>0$ (the case of negative $\omega_2$ is treated similarly) and $T>T_0=2\pi/\omega_2$. Without loss of generality one can consider $\omega_2=1$ and $T_0=2\pi$.
The general case can be reduced to this case by the time rescaling $t\mapsto \omega_2t$. 
If $\omega_1\notin \mathbb{Z}$ or $\omega_1\in\{-1,0\}$, consider the control $f_2(t)=\chi_{[0,2\pi]}(t)$ (the case $\omega_1=1$ corresponds to the harmonic system and is discussed in detail in  Appendix~\ref{C}). Then
\begin{align*}
&A^{2}_{13}\langle f_2\rangle
=v_{12}v_{23}\int_0^{2\pi}dt_1\int_0^{t_1}dt_2e^{-i\omega_1t_1-it_2}\\
&=\begin{cases}
-\dfrac{v_{12}v_{23}}{{\omega_1(\omega_1+1)}}(2\sin^2{\pi \omega_1}+i\sin{2\pi \omega_1})\neq 0,\quad \omega_1\notin \mathbb{Z}\\
-2v_{12}v_{23}\pi i\neq 0,\quad \omega_1=0 \quad \hphantom{-}[\textrm{Fig.}~\ref{Fig1:structure}(c)] \\
\hphantom{-}2v_{12}v_{23}\pi i \neq 0,\quad \omega_1=-1 \quad [\textrm{Fig.}~\ref{Fig1:structure}(e)]
\end{cases}
\end{align*}
If $\omega_1=n$, where $n\in \mathbb{Z}\setminus\{-1,0,1\}$, the form $A^{2}_{13}$ vanishes on the control $f_2$. In this case, we consider the nonconstant control $f_3(t)=(A\cos{nt}+B\sin{nt})\chi_{[0,2\pi]}(t)$, where $A^2+B^2\neq 0$.
Due to $n\in \mathbb{Z}\setminus\{-1,0,1\}$, we have $f_3\in \mathfrak{H}^{1}$. Then
\begin{eqnarray*}
A^{2}_{13}\langle f_3\rangle = v_{12}v_{23} \int_0^{2\pi}dt_1\int_0^{t_1}dt_2e^{-in t_1}e^{-it_2} (A\cos{nt_1}+B\sin{nt_1})(A\cos{nt_2}+B\sin{nt_2})\\
\quad=v_{12}v_{23}\frac{(iA+B)(A-iBn)\pi}{n^2-1}\neq 0.
\end{eqnarray*}

Thus for the anharmonic case there exists a nonzero control $f\in\mathfrak{H}^{1}$ such that $A^{2}_{13}\langle f\rangle\neq 0$ and, hence, $J_O^{(4)}(f_0)(f,\ldots,f)=\lambda_1|A^{2}_{13}\langle f\rangle|^2>0$.
\end{proof}

\begin{remark}
Note that if $\omega_1=-\omega_2$ then the quantum system is an uncontrollable $\Lambda$-atom~[Fig.~\ref{Fig1:structure}(f)] and $f_0\equiv 0$ is a third order trap.
\end{remark}

\section{Harmonic controllable system. Trap of the seventh order}
\label{C}

Consider the set $\mathfrak{H}^{3}$ of controls from $\mathfrak{H}^{1}$ on which the cubic form $A^{3}_{23}$ vanishes:
\begin{equation*}
\mathfrak{H}^{3}=\{f\in\mathfrak{H}^{1}\colon A^{3}_{23}\langle f\rangle=0\}.
\end{equation*}

\begin{theorem}
\label{thm2}
Let $(H_0,V)$ be a harmonic $\Lambda$-type system, i.e. such that $\omega_1=\omega_2=:\omega\neq 0$. Then the Taylor expansion of the objective functional $J_O$ at the constant zero control $f_0\equiv 0$ has the form
\begin{eqnarray*}
J_O(f_0&+&\delta f)=J_O(f_0)+\frac 12J_O^{(2)}(f_0)(\delta f,\delta f)\\&+&\frac 1{4!}J_O^{(4)}(f_0)(\delta f,\ldots,\delta f)\nonumber
+\frac 1{6!}J_O^{(6)}(f_0)(\delta f,\ldots,\delta f)\\&+&\frac 1{8!}J_O^{(8)}(f_0)(\delta f,\ldots,\delta f)+o(\|\delta f\|^{9}) \textrm{ as } \|\delta f\|\rightarrow 0\,
\end{eqnarray*}
where 
\begin{enumerate}
\item for any $\delta f\in \mathfrak{H}^0\setminus \mathfrak{H}^1$ one has $J_O^{(2)}(f_0)(\delta f,\delta f)<0$,
\item for any $\delta f\in \mathfrak{H}^1$ one has $J_O^{(2)}(f_0)(\delta f,\delta f)=0$ and
$J_O^{(4)}(f_0)(\delta f,\ldots,\delta f)=0$,
\item for any $\delta f\in \mathfrak{H}^1\setminus \mathfrak{H}^3$ one has $J_O^{(6)}(f_0)(\delta f,\ldots,\delta f)<0$,
\item for any $\delta f\in \mathfrak{H}^3$ one has $J_O^{(6)}(f_0)(\delta f,\ldots,\delta f)=0$ and
$J_O^{(8)}(f_0)(\delta f,\ldots,\delta f)
\geq0$.
\end{enumerate}
Moreover, if $|v_{12}|\neq|v_{23}|$ then for $T\geq T_0=2\pi/|\omega|$ there exists a control $\delta f\in \mathfrak{H}^3$ such that $J_O^{(8)}(f_0)(\delta f,\ldots,\delta f)>0$.
Thus, $f_0$ is a trap of the seventh order for the objective functional $J_O$. 
\end{theorem}
\begin{proof}
If $\omega_1=\omega_2=\omega$ [Fig.~\ref{Fig1:structure}(f)], then
$$
A^{2}_{13}\langle f\rangle
=\frac 12v_{12}v_{23}\left(\int_0^Tdte^{-\omega t}f(t)\right)^2.
$$
Hence $\frac 1{4!}J_O^{(4)}(f_0)(f,\ldots, f)=\lambda_1|A^{2}_{13}\langle f\rangle|^2=0$ for all $f\in \mathfrak{H}^1$ (for shortness of notation, in this proof we use $f$ instead of $\delta f$). 
Since $A^{1}_{23}\langle f \rangle=0$ and $A^{2}_{13}\langle f \rangle=0$ for all $f\in \mathfrak{H}^1$, expression~(\ref{J6}) implies that for all $f\in \mathfrak{H}^1\setminus\mathfrak{H}^3$ 
\begin{equation*}
\frac 1{6!}J_O^{(6)}(f_0)(f,\ldots, f)=\lambda_2|A^{3}_{23}\langle f\rangle|^2<0.
\end{equation*}
Additionally $A^{3}_{23}\langle f \rangle=0$ for all $f\in \mathfrak{H}^3$, so formulas~(\ref{J6}) and~(\ref{J8}) imply that for all $f\in \mathfrak{H}^3$ one has $\frac 1{6!}J_O^{(6)}(f_0)(f,\ldots, f)=0$ and
\begin{equation*}
\frac 1{8!}J_O^{(8)}(f_0)(f,\ldots, f)=\lambda_1|A^{4}_{13}\langle f\rangle|^2\geq 0.
\end{equation*}

Let us show that the set $\mathfrak{H}^{3}$ contains a control such that the value of the Fr\'{e}chet differential of the eighth order $J_O^{(8)}(f_0)$ on this control is positive. Consider $\omega>0$. The case of negative $\omega$ is treated similarly. Assume that $T\geq T_0=2\pi/\omega$.
Without loss of generality we consider $\omega=1$ and $T_0=2\pi$. The general case can be reduced to this case by the time rescaling $t\mapsto t/\omega$.
Then
\begin{align*}
A^{3}_{23}\langle f\rangle
=|v_{12}|^2v_{23}K_3\langle f\rangle+
|v_{23}|^2v_{23}R_3\langle f\rangle,
\end{align*}
where
\begin{eqnarray}
\label{K3}
K_3\langle f\rangle=\int_0^Tdt_1\int_0^{t_1}dt_2 \int_0^{t_2}dt_3e^{i(t_1-t_2-t_3)}f(t_1)f(t_2)f(t_3),\\
\label{R3}
R_3\langle f\rangle=\int_0^Tdt_1\int_0^{t_1}dt_2\int_0^{t_2}dt_3e^{i(-t_1+t_2-t_3)}f(t_1)f(t_2)f(t_3).
\end{eqnarray}

One could search for a nonzero control from $\mathfrak{H}^3$ in the general form as $f(t)=(A_0 +\sum_{k=2}^\infty (A_k \cos{k t}+B_k\sin{k t}))\chi_{[0,2\pi]}(t)$. As we will show below, it is sufficient to consider a nonzero control $f\in \mathfrak{H}^{3}$ in the form
\begin{equation}
\label{formspecialcontrol}
f(t)=(A+B\sin{2t} + C\cos{3t})\chi_{[0,2\pi]}(t)
\end{equation}
such that $K_3\langle f\rangle=0$ and $R_3\langle f\rangle=0$. Note that any control $f$ of the form~(\ref{formspecialcontrol}) belongs to $\mathfrak{H}^{1}$ because $\int_0^Te^{-it}f(t)dt=\int_0^{2\pi}e^{-it}f(t)dt=0$.
If we substitute control~(\ref{formspecialcontrol}) into the cubic forms~(\ref{K3}) and~(\ref{R3}), we get
\begin{eqnarray}
\label{K3ABC}
    K_3\langle f\rangle&=&\frac{1}{576} \pi  [12 C \left(-12 A^2+8 i A B+5 B^2\right)\nonumber +64 (3 A +2 i B) \left(6 A^2-B^2\right)\nonumber \\&&\quad-24 C^2 (3 A+2 i B)+9 C^3]
\end{eqnarray}
and
\begin{equation}
\label{R3=-2K3}
R_3\langle f\rangle=-2K_3\langle f\rangle.
\end{equation}
To find triplets $A,B,C$ such that $K_3\langle f\rangle=0$ and $R_3\langle f\rangle=0$, we need to equate to zero the real and imaginary parts of the right hand side of~(\ref{K3ABC}). As a result, we obtain the following system of two homogeneous cubic algebraic equations for the real variables $A,B,C$:
\begin{eqnarray}
\label{system}
\begin{cases}
384A^3-64AB^2+3C^3-24AC^2+C(20B^2-48A^2)=0, \\
48 A^2 B+6 A B C-8 B^3-3 B C^2=0.
\end{cases}
\end{eqnarray}
Obviously, nontrivial solutions may exist only for $A\ne 0$. 
Thus, taking $A=1$ we obtain 9 possible solutions $(A,B,C)$ of~(\ref{system}). For certainty, we take one of them:
\begin{equation}
A=1, \quad B = \frac{\sqrt{19-\sqrt{73}}}{\sqrt{6}}, \quad C = \frac{5+\sqrt{73}}{3}.
\label{parameters}
\end{equation}
For the control $f_4$ of the form~(\ref{formspecialcontrol}) with parameters~(\ref{parameters}), one has $R_3\langle f_4\rangle=-2K_3\langle f_4\rangle=0$. Hence, $A^{3}_{23}\langle f_4\rangle=0$
for any choice of matrix coefficients $v_{12}$ and $v_{23}$ of the free Hamiltonian $V$. This implies that $f_4\in \mathfrak{H}^{3}$.

Let us compute the value $J_O^{(8)}(f_0)(f_4,\ldots,f_4)$.
Note that
\begin{equation*}
    A^{4}_{13}\langle f\rangle
=|v_{12}|^2v_{12}v_{23}K_4\langle f\rangle+
|v_{23}|^2v_{12}v_{23}R_4\langle f\rangle,
\end{equation*}
where
\begin{eqnarray}
\label{K4}
K_4\langle f\rangle=\int_0^Tdt_1\int_0^{t_1}dt_2\int_0^{t_2}dt_3\int_0^{t_3}dt_4e^{i(-t_1+t_2-t_3-t_4)}\prod_{k=1}^4f(t_k),
\end{eqnarray}
\begin{eqnarray}
\label{R4}
R_4\langle f\rangle=\int_0^Tdt_1\int_0^{t_1}dt_2\int_0^{t_2}dt_3\int_0^{t_3}dt_4e^{i(-t_1-t_2+t_3-t_4)}\prod_{k=1}^4f(t_k).
\end{eqnarray}
If we substitute the control~(\ref{formspecialcontrol}) into the fourth-order forms~(\ref{K4}) and~(\ref{R4}), we get
\begin{eqnarray*}
    K_4\langle f\rangle &=& \frac{i \pi}{9216}(690 i A^3 + 27648 A^4 + 4600 i A^3 B  - 8448 i A B^3 - 16896 A^2 B^2+\\&& +2048 B^4 +C(1920 i B^3 - 230 A^3 - 192 A B^2) -27C^4 + C^3 (432 A + \\&&\quad +288 iB) + C^2 (432 B^2 - 1296 A^2 - 4608 i A B))
\end{eqnarray*}
and
\begin{equation}
\label{R_4=-K_4}
R_4\langle f\rangle=-K_4\langle f\rangle.
\end{equation}
So we can find the values of the forms $K_4$ and $R_4$ on the control $f_4$:
\begin{eqnarray*}
K_4\langle f_4 \rangle=-R_4\langle f_4\rangle = \frac{\pi}{10368}\left[(24 \sqrt{6} \left(5 \sqrt{73}+121\right)\right. \\ \left.\times \sqrt{19-\sqrt{73}}-i \left(887 \sqrt{73}+4603\right)\right]\neq 0.
\end{eqnarray*} 
If $|v_{12}|\neq|v_{23}|$ then 
$A^4_{13}\langle f_4\rangle \neq 0$ and 
\begin{equation*}
\frac 1{8!}J_O^{(8)}(f_0)(f_4,\ldots, f_4)=\lambda_1|A^{4}_{13}\langle f_4\rangle|^2>0.
\end{equation*}
Hence along the direction determined by the control $f_4$, the objective functional in the vicinity of zero control grows at eighth order.
\end{proof}

\begin{remark}
The values of the two different forms~(\ref{K4}) and~(\ref{R4}) computed on the same control $f_4$ differ only in sign. This is due to a special type of these forms and of the control.
For a general control $f\in\mathfrak{H}^0$, equality~(\ref{R_4=-K_4}) [and also equality~(\ref{R3=-2K3})] not necessarily will be satisfied.
\end{remark}

\section{Harmonic uncontrollable system. True traps}
\label{D}

As shown in the beginning of Appendix~\ref{A}, the general case of $\lambda_1>\lambda_3>\lambda_2$ can be reduced to the case $\lambda_1>0$, $\lambda_2<0$, $\lambda_3=0$. The formulas for the general case can be obtained from the reduced case using the observable transformation provided at the beginning of Appendix~\ref{A}. Thus in this Appendix we also consider the reduced case.

First we prove Proposition~\ref{MimMaxF1}. Any matrix $U\in \mathcal R$ can be parameterized by Euler angles $(\alpha,\beta,\gamma)$ as~\cite{ElovenkovaQuantumReport2023}:
\begin{multline}
\label{ZYZ}
U=U(\alpha,\beta,\gamma)=e^{-i\alpha J_Z}e^{-i\beta J_Y}e^{-i\gamma J_Z}=\\
\begin{pmatrix}
{e^{-i(\alpha+\gamma)}\cos^2{\beta/2}}& {-\frac{e^{-i\alpha}}{\sqrt{2}} \sin{\beta}}& {e^{-i(\alpha-\gamma)}\sin^2{\beta/2}}
\\ 
{\frac{e^{-i\gamma}}{\sqrt{2}}\sin{\beta}}& {\cos{\beta}}& -{\frac {e^{i\gamma}}{\sqrt{2}}\sin{\beta}}
\\
      {e^{i(\alpha-\gamma)}\sin^2{\frac\beta 2}}& {\frac{e^{i\alpha}}{\sqrt{2}} \sin{\beta}}& {e^{i(\alpha+\gamma)}\cos^2{\beta/2}}
\end{pmatrix}.
\end{multline}
Here $(\alpha,\beta,\gamma)\in (-\pi,\pi]\times[0,\pi]\times (-\pi,\pi]$. This decomposition is non-unique if $\beta=0,\pi$. 
Recall that
$
\rho_0=\mathrm{diag}(0,0,1)
$
and 
$
O=\mathrm{diag}(\lambda_1,\lambda_2,0)$, where $\lambda_1>0$, $\lambda_2<0$.
It turns out that the value of $F_1(U(\alpha,\beta,\gamma))$ depends only on the angle $\beta$:
\begin{multline*}
F_1(U(\alpha,\beta,\gamma))=
\mathrm{Tr}(OU(\alpha,\beta,\gamma)\rho_0U^\ast(\alpha,\beta,\gamma))=\frac {\lambda_1}4(1-\cos{\beta})^2+\frac {\lambda_2}2\sin^2{\beta}=:h(\beta).
\end{multline*}

The first and second derivatives of the function $h$ are:
\begin{eqnarray*}
h'(\beta)&=&\frac 12\sin{\beta}\left(\lambda_1-(\lambda_1-2\lambda_2)\cos{\beta}\right);\\
h''(\beta)&=&\frac 12 ((\lambda_1 - 2 \lambda_2) \sin^2{\beta}-(\lambda_1 - 2 \lambda_2) \cos^2{\beta}+\lambda_1 \cos\beta ).
\end{eqnarray*}
The point $\beta=\pi$ is a point of global maximum of the function $h$. The value of this global maximum is $\lambda_1$ and this value coincides with the global maximum of the function $F_1$ on the group $\mathcal{R}$. The point $\beta=\arccos\left(\frac{\lambda_1}{\lambda_1-2\lambda_2}\right)$ is a point of global minimum of the function $h$. The value of this global minimum is $\frac{\lambda^2_2}{2\lambda_2-\lambda_1}$ and this value coincides with the global minimum of the function $F_1$ on the group $\mathcal{R}$.
Since $h'(0)=0$, $h''(0)=2\lambda_2<0$ and $h(0)=0<\lambda_1$, the point $\beta=0$ is a local but not global maximum point for the function $h$. If we substitute $\beta=0$ in~(\ref{ZYZ}) we get the matrices
$$
e^{-i\phi J_Z}=
\begin{pmatrix}
e^{-i\phi}&0&0\\
0&1&0\\
0&0&e^{i\phi}
\end{pmatrix},
$$
where $\phi\in[-\pi,\pi)$. These matrices form the group $\mathcal{R}_Z$. Thus $U\in \mathcal{R}$ is a point of local but not global (not strict) maximum of the function $F_1$ if and only if $U\in \mathcal{R}_Z$. Proposition~\ref{MimMaxF1} is proved. 

Now we prove Theorem~\ref{Thmuncontrollablecase}. From Proposition~\ref{MimMaxF1} it follows that for $T\geq T_0$ controls $f\in L_2([0,T]),\mathbb{R})$ such that
$U_T^f\in \mathcal{R}_Z$, are traps in the sense of Definition~\ref{DefTraps}.
In addition to such controls, only singular controls could be traps. It is known that for controlling of single isolated qubit, there is only one singular control, which is constant~\cite{PechenPRA2012}. Since $\mathfrak{Lie}(iH_0,iV)\cong \mathfrak{su}(2)$, the map $L_2([0,T],\mathbb{R})\ni f\mapsto U_T^f\in \mathcal{R}$ also has a single singular control, which is the constant zero control $f_0\equiv 0$.
Since $U_T^{f_0}=\mathbb I\in \mathcal{R}_Z$, this control is a true trap. Theorem~\ref{Thmuncontrollablecase} is proved.

\section*{Acknowledgements}
The authors are grateful to Anton S.~Trushechkin for discussing the role of the choice of initial controls on the optimization, to Sergei A.~Kuznetsov for help in summarizing the controllability results, and to Maria A.~Elovenkova for discussing kinematic control landscape for the uncontrollable harmonic $\Xi$ system. The research for the uncontrollable case (Section~V and Appendix D) was supported in part by the federal academic leadership program ‘‘Priority 2030’’ (MISIS Strategic Project Quantum Internet). The research in Section~III,~IV,~VI and in Appendices A, B, C was supported by the RSF grant no. 22-11-00330, \url{https://rscf.ru/en/project/22-11-00330/}, and performed at Steklov Mathematical Institute of Russian Academy of Sciences.

\end{document}